\documentclass[letterpaper, DIV=18]{scrartcl}

\usepackage{amssymb, amsthm, mathtools,bbm,bm}
\usepackage[square,sort,comma,numbers]{natbib}
\usepackage{enumitem,color}

\usepackage[none]{hyphenat}
\newtheorem{theorem}{Theorem}

\newtheorem{lemma}[theorem]{Lemma}
\newtheorem{definition}[theorem]{Definition}
\newtheorem{proposition}[theorem]{Proposition}

\numberwithin{equation}{section}
\numberwithin{theorem}{section}
\theoremstyle{remark}

\newcommand{\diam}{\operatorname{diam}}

\newcommand{\tr}{{\operatorname{Tr}\,}}

\newcommand{\beq}[1]{\begin{equation} \label{#1}}
	\newcommand{\eeq}{\end{equation}}

\newcommand{\arcosh}{\operatorname{arcosh}}
\renewcommand{\epsilon}{\varepsilon}

\def\be{\begin{equation}}
	\def\ee{\end{equation}} 

\DeclareMathOperator{\dist}{dist}

\newcommand{\eps}{\epsilon}

\begin{document}
	\addtokomafont{author}{\centering}
	
\title{ \centering  The Quantum Random Energy Model \\  is the Limit of Quantum $ p $-Spin Glasses}

 	\author{  Anouar Kouraich, Chokri Manai,  Simone Warzel}
	\date{ \small December 19, 2025}
	
	\maketitle

	\minisec{Abstract}
	We consider the free energy of a class of spin glass models with $ p$-spin interactions in a transverse magnetic field. As $p \to \infty$, the infinite system-size free energy is proven to converge to that of the quantum random energy model. This is accomplished by combining existing analytical techniques addressing the non-commutative properties of such quantum glasses, with the description of the typical geometry of extreme negative deviations of the classical $ p $-spin glass.  We also review properties of the corresponding classical free energy and conjectures addressing $ 1/p $-corrections in the quantum case.\\
	

	\bigskip	
	

\section{Introduction} 
Since the advent of classical, mean-field spin glass models, the interest in the influence of quantum effects on the physical properties of such systems has never ceased~\cite{FS86,Goldschmidt:1990aa,TD90,Cesare:1996aa,Obuchi:2007aa} (see also~\cite{CM22}).  
More recently, several rigorous results have been established, mainly concerning the simplest class of such quantum models namely the transversal field models. These results cover explicit formulas for the pressure (or equivalently: free energy) of generalized random energy models~\cite{MW20,MW21,MW22,Sh24}, the Hopfield model~\cite{Shcherbina:2020aa}, as well as a rather general quantum Parisi formula \cite{MW25,AdBr20}, and a proof of the persistence of replica symmetry breaking~\cite{Leschke:2021ab,Itoi:2023ac}.  

The main aim of this note is to anchor the quantum random energy model within the class of quantum $ p$-spin models. We prove the analogue of a classical result on the approximation of its pressure as the limit of $ p$-spin models, thus completing a task that was left open in~\cite{Manai:2020aa}.

\subsection{Classical case}
Classical Ising spin-glasses are  Gaussian processes on the configuration space $ \mathcal{Q}_N := \{ -1, 1\}^N $ with $ N $ the number of Ising spins. Any realization
$ U_p:   \mathcal{Q}_N  \to \mathbb{R} $ of such a  process represents a random energy landscape on $ \mathcal{Q}_N $. The process is uniquely characterized by its mean, which we will assume to be zero, $ \mathbb{E}[U_p(\pmb{\sigma}) ] = 0 $ for all $ \pmb{\sigma} = (\sigma_1 , \dots , \sigma_N) \in  \mathcal{Q}_N $, and its covariance,
\begin{equation}\label{eq:cov}
 \mathbb{E}[U_p(\pmb{\sigma}) U_p(\pmb{\tau}) ] = N  \  c_{p,N}\left( r_N(\pmb{\sigma}, \pmb{\tau})\right) ,
 \end{equation}
 which is a function $  c_{p,N}: [-1,1] \to  [-1,1] $ of the overlap 
$$
  r_N(\pmb{\sigma}, \pmb{\tau}) = \frac{1}{N} \sum_{j=1}^N \sigma_j \tau_j 
$$
 of two spin configurations. Throughout this note, we will assume that, for the fixed $ p \geq 2 $ under consideration: 
\begin{equation}\label{eq:ass}
\lim_{N\to \infty} \sup_{x\in [-1,1] } \left| c_{p,N}(x) - x^p \right| = 0 . 
\end{equation}
This covers any $ p $-spin interaction of the form
\begin{equation}\label{eq:pspin}
U_p(\pmb{\sigma}) =  \sqrt{\frac{p!}{N^{p-1}} } \sum_{ 1 \leq j_1< \dots <  j_p \leq N }^N g_{j_1,\dots,j_p}{\sigma}_{j_1} \cdots {\sigma}_{j_p}
\end{equation}
whose $ p $-body couplings are random arrays  $g_{j_1,\dots,j_p} $  of independent and identically distributed (i.i.d.), centered Gaussian random variables with variance one, as well as its simpler variant in which the summation in~\eqref{eq:pspin} extends to all indices $ j_1, \dots, j_p  \in \{ 1, \dots , N \}  $ and the normalization factor is changed accordingly to $ N^{\frac{1-p}{2}}$.  The latter case corresponds to $ c_{p,N}(x) = x^p $, independent of $ N $. 

The classical pressure
\begin{equation}\label{eq:clpressure}
 \Phi_{p,N}(\beta, 0) := \frac{1}{N} \ln \frac{1}{ 2^N}\,  \sum_{\pmb{\sigma}\in \mathcal{Q}_N} \exp\left( - \beta U_{p}(\pmb{\sigma}) \right)
\end{equation}
describes the thermodynamic properties of such glasses as a function of the inverse temperature $ \beta > 0 $. In the limit  $ N\to\infty  $ this function asymptotically agrees with its probabilistic (so-called: quenched) average, cf.~\eqref{eq:selfaveraging} below. 
Dazzling characteristics such as replica-symmetry breaking at low temperatures are captured by Parisi's infamous variational formula ~\cite{Mezard:1986aa,Tal11a,Tal11b,Pan13} for the quenched limit,  $\mathbb{E}\left[ \Phi_{p}(\beta, 0) \right] =\lim_{N\to \infty} \mathbb{E}\left[ \Phi_{p, N}(\beta, 0) \right]$. 
Despite being less complex than the original many-body problem, this variational formula is far from being simple, and fine properties of its optimizer, Parisi's replica-order parameter, are still under active investigation for general $ p$-spin glasses~\cite{Auffinger:2015aa,YZhou24}. It has been a comforting fact that the pressure has been explicitly computed in the large $ p $-limit. More precisely, in case $ p = \infty $, where the Gaussian process $ U_\infty(\pmb{\sigma} ) = \sqrt{N} g(\pmb{\sigma}) $ is composed of i.i.d. Gaussian random variables $  g(\pmb{\sigma})  $ with variance one, the (quenched) pressure is asymptotically given by~\cite{Derrida:1981aa,Bov06}
\begin{equation}
\lim_{N\to \infty}\mathbb{E}\left[ \Phi_{\infty,N}(\beta,0) \right] = \Phi_{\infty}(\beta,0)  = \left\{ \begin{array}{ll} \frac{1}{2}\beta^2   , & \beta \leq \beta_c,  \\
 									 \beta \beta_c - \frac{1}{2} \beta_c^2   , & \beta > \beta_c . \end{array} \right.
\end{equation}
In this case, a (second-order) phase transition at $ \beta_c := \sqrt{2\ln 2}  $ towards a glass phase is evident.  
One motivation for the study of this limiting case, known as the random energy model (REM), was the fact that the infinite system-size pressure is continuous as $ p \to \infty $. 
\begin{proposition}[cf.~\cite{Talagrand:2000aa,Bovier:2002aa,Panchenko:2014aa}]\label{prop:Clpinfty}
For any $ \beta > 0 $:
\begin{equation}\label{eq:convclassical}
 \lim_{p\to\infty} \lim_{N\to \infty} \mathbb{E}\left[ \Phi_{p,N}(\beta, 0) \right] =   \Phi_{\infty}(\beta,0)  .
 \end{equation} 
 \end{proposition}
 This result dates back to Derrida~\cite{Derrida:1980aa}. Its rigorous proof in the present setting is contained in the following works. 
The existence of the  limit $ N \to \infty $ is implied by Panchenko's proof of the Parisi formula for $ p $-spin glasses~\cite{Panchenko:2014aa}. 
The convergence of the subsequent limit $ p \to \infty $ is implicitly contained in \cite{Talagrand:2000aa,Bovier:2002aa}. It follows straightforwardly from bounding the derivative   $  \frac{\partial}{\partial \beta} \Phi_{p,N}(\beta,0) $  from above by $ \beta_c $ using the maximal inequality \cite[Thm. 3.5]{BLM13},  
$$\mathbb{E}[\min_{ \pmb{\sigma} } U_p( \pmb{\sigma}) ]   \geq - \beta_c N .
$$
Moreover, by the (truncated) second moment method the annealed and quenched pressure agree in an interval enlarging with $ p $, i.e. $  \Phi_{p}(\beta,0)  = \beta^2/2 $ for $ \beta \leq \beta_c ( 1-c_p ) $ with $ \lim_{p\to \infty} c_p = 0 $, cf.~\cite{Bov06}.  Convexity of the pressure $   \Phi_{p}(\beta,0)$ as a function of $ \beta $, then yields~\eqref{eq:convclassical}.  \\

\subsection{Quantum case} 
The main purpose of the present note is to generalize Proposition~\ref{prop:Clpinfty}  to the situation in which an additional transversal constant magnetic field of strength $ \Gamma > 0 $ is present. Instead of Ising spins, the basic entities are then $ N $ quantum spin-$1/2 $, whose Hilbert space can be represented as $ \ell^2(\mathcal{Q}_N) $, i.e., square-summable sequences indexed by Ising configurations $ \pmb{\sigma} $. The latter define an orthonormal basis $ | \pmb{\sigma}\rangle $, $\pmb{\sigma} \in \mathcal{Q}_N $, in this Hilbert space. Here and in the following, we employ the bra-ket notation, in which $ \langle \pmb{\sigma} | \psi \rangle $ represents the scalar product of such basis vectors with an arbitrary vector $ \psi $ on the Hilbert space.  In the distinguished orthonormal basis, the random, self-adjoint Hamiltonian
$$
 H_{p,N} = U_p - \Gamma T 
$$
is composed of a self-adjoint multiplication operator, which is diagonal,  $ U_p | \pmb{\sigma} \rangle = U_p(\pmb{\sigma})| \pmb{\sigma} \rangle $, and the operator involving the transversal field of strength $ \Gamma > 0 $, which acts as 
\begin{equation}
\langle \pmb{\sigma} |  T \psi\rangle= \sum_{j=1}^N  \langle F_j \pmb{\sigma} | \psi\rangle \qquad \mbox{on}\quad \psi \in  \ell^2(\mathcal{Q}_N)  ,
\end{equation}
with $ F_j \pmb{\sigma} := (\sigma_1, \dots , - \sigma_j , \dots , \sigma_N) $ the spin-flip operator on the $j$th component. 

%

Quantum mechanically, to obtain the pressure, one substitutes in~\eqref{eq:clpressure} the sum on configurations by a trace,
\begin{equation}
 \Phi_{p,N}(\beta, \Gamma) = \frac{1}{N} \ln \frac{1}{ 2^N}\,  \tr \exp\left( - \beta H_{p,N}\right)  . 
\end{equation}
If $ \Gamma = 0 $ this reduces to~\eqref{eq:clpressure}. The description of the limit $ N \to \infty$ through a quantum Parisi formula is available~\cite{MW25} for mixed spin-glasses -- at least for even $ p$. 
However, the variational principle is even more cumbersome than in the classical case. Remarkably, for the case of the quantum REM ($ p= \infty $), the limiting quenched pressure can still be computed~\cite{Goldschmidt:1990aa}
\begin{equation}\label{eq:PressureQREM}
\Phi_\infty(\beta,\Gamma) =  \max  \left\{ \Phi_\infty(\beta,0), \ln \cosh (\beta \Gamma )\right\} .
\end{equation}
A proof of this formula was accomplished in \cite{MW20} and the method was further used to characterize the large deviations of trajectories in a REM potential \cite{GMW23,MW25b}.
The quantum feature is a (first-order) phase transition at the critical field strength
$$
 \Gamma_c(\beta) := \beta^{-1} \arcosh\left( \exp\left(\Phi_\infty(\beta,0)\right)  \right) 
$$
into a quantum paramagnetic phase. In contrast to the case of a longitudinal field, whose critical field strength is known as the Almeida-Thouless line, the quantum paramagnetic phase extends even to zero temperature ($\beta = \infty$). \\

Our main result is the continuity of the quenched pressure in the limit $ p \to \infty $. This completes a partial result in~\cite{Manai:2020aa}, in which a coupled limit $ p(N) \to \infty $ was considered. 
\begin{theorem}\label{thm:main}
For any $\beta , \Gamma \geq 0 $: 
\begin{equation}\label{eq:main}
\lim_{p\to \infty}  \liminf_{N\to \infty} \mathbb{E}\left[  \Phi_{p,N}(\beta, \Gamma) \right]  =  \lim_{p\to \infty}  \limsup_{N\to \infty} \mathbb{E}\left[  \Phi_{p,N}(\beta, \Gamma) \right]  =\Phi_\infty(\beta, \Gamma) 
\end{equation}
\end{theorem}
The proof, which combines functional analytic techniques from \cite{MW20,Manai:2023ys} with the probabilistic control of the size of clusters of extreme negative energies of $ U_p $,  will be spelled out in Section~\ref{sec:proof}.

Several remarks are in order:
\begin{enumerate}
\item 
In case of a pure $ p $-spin glass with $ p $ even, the existence of the limit $   \Phi_{p}(\beta, \Gamma) =  \lim_{N\to \infty} \mathbb{E}\left[  \Phi_{p,N}(\beta, \Gamma) \right]  $ is guaranteed as a by-product of the proof of the quantum Parisi formula in~\cite{MW25}. In this case, one may reformulate the result as $   \lim_{p\to \infty}  \lim_{N\to \infty} \mathbb{E}\left[  \Phi_{2p,N}(\beta, \Gamma) \right]  = \Phi_\infty(\beta, \Gamma) $. 
Since in the general case covered here by assumption~\eqref{eq:ass}, the existence of the limit $ N \to \infty $ is not warranted, we use the upper and lower limits in~\eqref{eq:main}. 

\item 
Based on non-rigorous calculations using the replica trick and a $ 1/p $ expansion \cite{Goldschmidt:1990aa,TD90,Cesare:1996aa,Obuchi:2007aa}, physicists long predicted the phase diagram of the quantum $ p $-spin glass	to converge to that of the quantum REM. 
It is an interesting question to justify these $ 1/p $ corrections. On the level of the quenched pressure, these predictions~\cite{TD90,Cesare:1996aa} agree with a (non-rigorous) calculation up to second-order perturbation theory:
\begin{equation}
 \Phi_p(\beta,\Gamma) \approx \Phi_\infty(\beta,\Gamma) +\frac{1}{p} \left\{ \begin{array}{ll} \frac{\beta}{2\Gamma \tanh(\beta \Gamma) } , &\Gamma >  \Gamma_c(\beta) , \\[1ex]
 \frac{\Gamma^2}{2}  , & \Gamma <  \Gamma_c(\beta) , \; \beta < \beta_c , \\[1ex]
 \frac{\Gamma^2 \beta  }{2\beta_c} , & \Gamma <  \Gamma_c(\beta) , \;  \beta > \beta_c . 
 \end{array} 
  \right.
\end{equation}
In the quantum paramagnetic phase, characterized by $ \Gamma >  \Gamma_c(\beta) $,  the unperturbed system is thereby taken to be the quantum paramagnet $ - \Gamma T $ and the perturbation is the REM. In the unfrozen and frozen REM, characterized by $ \Gamma <  \Gamma_c(\beta) $ and  $ \beta < \beta_c $ and $ \beta > \beta_c $,  the roles are exchanged in this second-order calculation. 
Establishing at least part of this rigorously would be a major step and require to extend the analysis in~\cite{Manai:2023ys} from $ p = \infty $ to finite, but large $ p $.
\item

In the proof of Theorem~\ref{thm:main}, use will be made of the self-averaging property of the pressure, which is known to extend straightforwardly from the classical to the quantum case, i.e. for all $ t > 0 $:
\begin{equation}\label{eq:selfaveraging}
  \mathbb{P}\left(\left| \Phi_N^{p}(\beta,\Gamma)  - \mathbb{E}\left[\Phi_N^{p}(\beta,\Gamma) \right]  \right|  > \frac{t \, \beta }{\sqrt{N} } \right) \leq  2 \,  \exp\left(- \frac{t^2}{4}\right) , 
\end{equation}
cf.~\cite{Crawford:2007aa} and \cite[Prop. 2.1]{Manai:2020aa}. 
\end{enumerate}

\section{Proof of the main result}\label{sec:proof}

Theorem~\ref{thm:main} follows by establishing asymptotically coinciding upper and lower bounds. 
\subsection{Lower bound}
A lower bound has essentially been established in~\cite{MW20,Manai:2020aa}. In order to keep this note self-contained, we will spell out a sketch of the proof.%
\begin{proposition}[cf.\ Lemma 2.1 in~\cite{MW20}]\label{prop:lower}  For any $\beta , \Gamma \geq 0 $: 
\begin{equation}
	\lim_{p\to \infty} \liminf_{N\to \infty} \mathbb{E}\left[ \Phi_{p,N}(\beta, \Gamma) \right]\geq  \max\left\{ \Phi_{\infty}(\beta,0)  , \ln \cosh (\beta \Gamma ) \right\} .
	\end{equation}
\end{proposition}
\begin{proof} 
The proof is based on the Gibbs variational principle,
\begin{equation}\label{eq:Gibbs}
	\ln \tr e^{-\beta H} = -  \inf_{\varrho} \left[  \beta \, \tr \left(H \varrho \right) +  \tr\left( \varrho \ln \varrho \right) \right] , 
	\end{equation}
	in which the infimum is taken over all density matrices, i.e.\ positive-definite $ \varrho \geq 0 $, with unit trace $ \tr \varrho = 1 $, on $ \ell^2(\mathcal{Q}_N) $. 
Inserting the two canonical choices:  (i)~the classical Gibbs state of $ p $-spin interaction, $ \varrho \propto e^{-\beta U_p} $,  and  (ii)~the Gibbs state of the quantum paramagnet $ \varrho \propto e^{\beta \Gamma T}   $, and taking expectations yields
$$
	\mathbb{E}\left[ \Phi_{p,N}(\beta, \Gamma) \right] \geq  \max\left\{\mathbb{E}\left[ \Phi_{p,N}(\beta,0) \right] ,  \ln \cosh (\beta \Gamma ) \right\} .
$$
Taking the limit $ N \to \infty $ and subsequently $ p \to \infty $,  the classical quenched pressure converges by~\eqref{eq:convclassical} to the pressure of the REM. This completes the proof. 
\end{proof}

\subsection{Geometry of extreme negative deviations}
It thus remains to establish an upper bound. To do so, we decompose the Hamming cube $\mathcal{Q}_N = \{ -1,1 \}^N $ into sites of extreme negative deviation 
\begin{equation} 
\mathcal{L}_\varepsilon := \left\{ \pmb{\sigma} \in \mathcal{Q}_N  \ | \  U_p(\pmb{\sigma}) < - \varepsilon N \right\}  ,
\end{equation}
and its complement with  $\varepsilon > 0 $ as a variational parameter. Our basic strategy is to remove the restriction of $ T $ to the entire $1 $-step augmented region 
$$
  \mathcal{L}_\varepsilon^+ :=  \left\{ \pmb{\sigma} \in \mathcal{Q}_N \ | \  \dist( \pmb{\sigma} ,   \mathcal{L}_\varepsilon ) \leq 1 \right\} 
$$
from the Hamiltonian. Here and in the following, $\dist(\cdot, \cdot)$ refers to the canonical Hamming metric $\dist(\pmb{\sigma} , \pmb{\sigma}' ) := \sum_{i = 1}^{N} \frac12 |\sigma_i - \sigma'_i|$ for $\pmb{\sigma} , \pmb{\sigma}' \in \mathcal{Q}_N$. Let $  T_{ \mathcal{L}_\varepsilon^+ } $ stand for the self-adjoint operator on $ \ell^2(\mathcal{Q}_N) $ defined through its matrix elements 
\begin{align}
\langle \pmb{\sigma} |  T_{ \mathcal{L}_\varepsilon^+ } |  \pmb{\sigma}' \rangle :=  \begin{cases} 1 & \mbox{in case $  \pmb{\sigma} , \pmb{\sigma}'  \in  \mathcal{L}_\varepsilon^+  $ with $ \dist(\pmb{\sigma} , \pmb{\sigma}' ) = 1 $,} \\ 0 & \mbox{else.}
 \end{cases} 
\end{align} 
Up to this perturbation, the Hamiltonian is thus a direct sum
\begin{equation}\label{eq:directsumH}
H_{p,N} = U_p 1_{\mathcal{L}_\varepsilon} \oplus H_{\mathcal{L}_\varepsilon^c }    -  \Gamma T_{ \mathcal{L}_\varepsilon^+ } , 
\end{equation}
involving the multiplication operator by $ U_p $ on $ \ell^2(\mathcal{L}_\varepsilon ) $ and the restriction of $ H_{p,N}  $ to $ \ell^2(\mathcal{L}_\varepsilon^c ) $, which we abbreviate by 
$ H_{\mathcal{L}_\varepsilon^c } $. In order to estimate the operator norm of $ T_{ \mathcal{L}_\varepsilon^+ }  $, we use a similar strategy as in~\cite{MW20} and cover $ \mathcal{L}_\varepsilon^+ $ by a union of connected clusters. In contrast to the case $ p = \infty $, for which the correlation length is one, we need to consider the effect of an extensive correlation length.  
It is therefore reasonable to identify sites in  $  \mathcal{L}_\varepsilon^+ $, which are at a distance at most $Nr /2 $, with some $ r$ (to be chosen later as the correlation length). 
\begin{definition}
	Let $ r \in (0,1) $. 
    We call a set $C\subset \mathcal{Q}_N$ \emph{$r$-connected} if for any $ \pmb{\sigma}, \pmb{\sigma}^\prime\in C$ there exists a sequence $ \pmb{\sigma}= \pmb{\sigma}^0, \pmb{\sigma}^1,\dots, \pmb{\sigma}^m= \pmb{\sigma}^\prime$ all in $\mathcal{L}_\varepsilon^+ $ such that $\dist(\pmb{\sigma}^i,\pmb{\sigma}^{i+1}) < Nr/2 $ for all $0\leq i \leq m-1$. We call $C\subset \mathcal{L}_\varepsilon^+ $ a \emph{maximal $r$-connected component} if $C$ is $r$-connected and for any $r$-connected $C^\prime$ with $ C\subset C^\prime\subset \mathcal{L}_\varepsilon^+ $, it follows that $C=C^\prime$. We denote the family of maximal $r$-connected components of $\mathcal{L}_\varepsilon^+ $ by $\mathcal{C}_{\eps,r}$.
\end{definition}
By construction, the augmented extreme deviation set decomposes into disjoint maximal $r$-connected components
\begin{equation}\label{eq:decompcluster}
\mathcal{L}_\varepsilon^+= \biguplus_{C \in \mathcal{C}_{\varepsilon, r} } C . 
\end{equation}
Our construction relies on a cut-off in energies. 
It is the simple, natural generalization of the decomposition of the extremal sites for the REM in~\cite{MW20}. In contrast to the REM, however, one cannot expect the extremal deviations to be isolated; rather, one needs to group them in clusters of balls with the radius of the correlation length.  This is related to the overlap-gap property for large $ p $, which is used in the construction of 'lumps' and the investigation of the shattering transition for $ p $-spin models, see e.g.~\cite{Talagrand:2000aa,GJK23,Ala24}. This picture also suggests that the distribution of the (Q)REM's random overlap converges to that of the (Q)REM in the limit $ p \to \infty $. A proof of this would require an additional analysis.  

The main technical result in the proof is the control of the diameter of the maximal $r$-connected components. 
This is due to the bad localization property of $ T $, when restricted to any closed ball 
$$ \overline{B}_{Nr} \equiv  \overline{B}_{Nr}(\pmb{\sigma}^0)  := \{  \pmb{\sigma} \in \mathcal{Q}_N \ | \ \dist( \pmb{\sigma},  \pmb{\sigma}^0) \leq Nr \} $$ centered about some $ \pmb{\sigma}^0 $. Bounds on the operator norm of $ T_{\overline{B_{rN}}} $ were established in~\cite{FT05} with a minor refinement in~\cite{Manai:2023ys}. We reformulate them for our purpose:
\begin{lemma}\label{lem:opnormT}
 If $\displaystyle  \max_{C \in  \mathcal{C}_{\varepsilon, r} } \diam\ C \leq N r L $ with $ 0 < r L < \frac{1}{2} $, then for all $ N > (rL)^{-1} $:
 $$ \left\| T_{ \mathcal{L}_\varepsilon^+ }  \right\| \leq  \left\| T_{\overline{B}_{NrL}} \right\| \leq 2 N \sqrt{rL} .
 $$ 
\end{lemma}
\begin{proof}
Thanks to~\eqref{eq:decompcluster}, we have
$
 \left\| T_{ \mathcal{L}_\varepsilon^+ }  \right\| \ =  \max_{C \in  \mathcal{C}_{\varepsilon, r} } \left\| T_{C}  \right\| $. 
By assumption, any $ C \in  \mathcal{C}_{\varepsilon, r} $ is contained in a ball $ B_{NrL} $ of radius $ N r L $. 
Since the restriction of $ T $ to any subset of $ \mathcal{Q}_N $ is a positivity improving operator, its operator norm is monotone increasing in the domain. 
Hence $  \left\| T_{C}  \right\| \leq  \left\| T_{ \overline{B}_{NrL}}  \right\| \leq 2 N \sqrt{ rL ( 1 - rL + N^{-1} )} $, where the last inequality is \cite[Prop. 3.1]{Manai:2023ys}. 
\end{proof}

The diameter of any maximal $r$-connected component of $\mathcal{L}_\varepsilon^+ $ is, in turn, controlled probabilistically.

\begin{lemma}\label{lem:prob}
For any $ \varepsilon > 0 $ and all $ p $ large enough, there is some $ c_p(\varepsilon) > 0  $ and $ r \equiv r_p(\varepsilon) \leq 1 $, $ L \equiv L_p(\varepsilon)\in \mathbbm{N} $ such that:
\begin{enumerate}
\item for all sufficiently large $ N $: \quad $ \displaystyle 
	 \mathbb{P}\left(   \max_{C \in  \mathcal{C}_{\varepsilon, r} } \diam\ C  >  N r L   \right) \leq e^{-N c_p(\varepsilon)}  $, 
\item 
 $\displaystyle  \lim_{p\to \infty} L_p(\varepsilon)  r_p(\varepsilon) = 0 $. 
 \end{enumerate}
\end{lemma}
\begin{proof}
The proof starts from the observation that in case there is a maximally $ r $-connected subset $C\subset \mathcal{L}_\varepsilon^+ $ of diameter exceeding $ N r L > 0 $,  then there is a path of at least $ L $ sites $ \pmb{\sigma}^0 , \pmb{\sigma}^1, \dots , \pmb{\sigma}^{L-1} \in  \mathcal{L}_\varepsilon  $ such that:
\begin{enumerate}
\item at step $ j = 1, \dots , L $, one connects $ \pmb{\sigma}^{j-1} $ to $ \pmb{\sigma}^j $ with $ \dist( \pmb{\sigma}^{j-1} ,  \pmb{\sigma}^j ) \in [N r/2 , N r]$. 
\item
the path is strongly self-avoiding in the sense that $   \pmb{\sigma}^j  \not\in \bigcup_{k=0}^{j-1} B_{Nr/2}(\pmb{\sigma}^k) $ with open balls $ B_{Nr/2}(\pmb{\sigma}^k)  := \{ \pmb{\sigma} \ | \  \dist(\pmb{\sigma}^k,  \pmb{\sigma}  ) < Nr /2 \} $.
\end{enumerate}
To construct such a path, we employ a last exit algorithm. By the definition of $ r $-connectedness and since $ \diam C \geq NLr +1 $, there are two sites $  \pmb{\tau}^0 ,  \pmb{\tau}^M \in \mathcal{L}_\varepsilon $ with  $ \dist(\pmb{\tau}^0 ,  \pmb{\tau}^M ) \geq N rL -1 $ . These points serve as the extremal points of a path  of extremal sites $ \pmb{\tau}^0,  \pmb{\tau}^1 , \dots \pmb{\tau}^{M-1}, \pmb{\tau}^{M} \in \mathcal{L}_\varepsilon  $ such that $ \dist(\pmb{\tau}^{j-1}, \pmb{\tau}^{j}) \in [1,Nr/2)$.  Without loss of generality, one may assume that this path of sites is already self-avoiding.  
From this self-avoiding path, we then construct $ \pmb{\sigma}^0 , \pmb{\sigma}^1, \dots , \pmb{\sigma}^{L-1}  $ by thinning and a last exit strategy: 
\begin{enumerate}
\item we set $  \pmb{\sigma}^0 :=  \pmb{\tau}^0 $ and select $ \pmb{\sigma}^1 $ from $ \pmb{\tau}^1 , \dots \pmb{\tau}^{M-1}, \pmb{\tau}^{M} $ as the last exit from the annulus
$$ A_r( \pmb{\sigma}^{0}) := \{ \pmb{\sigma}  \ | \ Nr/2 \leq \dist( \pmb{\sigma}^{0}, \pmb{\sigma} ) \leq Nr \} . $$ 
\item at the step from $  j = 1, \dots , L-2 $ to $ j+1 $, we iteratively select $ \pmb{\sigma}^{j+1} $ from the remaining path $  \pmb{\sigma}^{j} = \pmb{\tau}^{k(j)} , 
 \pmb{\tau}^{k(j)+1} , \dots , \pmb{\tau}^M $ as the last exit point from the union $   \bigcup_{k=0}^{j} B_{Nr}(\pmb{\sigma}^k)  $ in the annulus $ A_r( \pmb{\sigma}^{j}) $. Note that such a point exists, since the remaining path is by construction already avoiding $    \bigcup_{k=0}^{j-1} B_{Nr}(\pmb{\sigma}^k)  $, and, by $ r $-connectedness, there needs to be a vertex in the annulus $ A_r( \pmb{\sigma}^{j}) $. 
 \end{enumerate}
 This path has the desired properties listed above as 1.-2. Moreover, since by construction
 $$
  \dist(\pmb{\sigma}^0, \pmb{\sigma}^{L-1}) \leq \sum_{j=1}^{L-1}   \dist(\pmb{\sigma}^{j-1}, \pmb{\sigma}^{j})  \leq Nr (L-1) \leq  NLr - N < \dist(\pmb{\tau}^0 ,  \pmb{\tau}^M )  ,
  $$
   we indeed find at least $ L $ sites in this manner.

 Therefore, we may use a union bound to estimate
\begin{align}\label{eq:unionb}
 & \mathbb{P}\left(   \max_{C \in  \mathcal{C}_{\varepsilon, r} } \diam\ C  >  N r L  \right) 
   \leq 
 \mathbb{P}\left( \begin{array}{c} \mbox{There are $ \pmb{\sigma}^0 , \pmb{\sigma}^1, \dots , \pmb{\sigma}^{L-1} \in  \mathcal{L}_\varepsilon $ with} \\    \mbox{$ \dist( \pmb{\sigma}^{j-1},  \pmb{\sigma}^j) \leq N r $ and  $  \pmb{\sigma}^j \not\in \bigcup_{k=0}^{j-1} B_{\frac{rN}{2}}( \pmb{\sigma}^k ) $}\end{array} \right) \notag \\
 & \leq \sum_{\pmb{\sigma}^0 \in \mathcal{Q}_N} \sum_{\pmb{\sigma}^1 \in A_r( \pmb{\sigma}^{0}) \backslash B_{Nr/2}(\pmb{\sigma}^0) } \dots \mkern-10mu \sum_{ \pmb{\sigma}^{L-1} \in  A_r( \pmb{\sigma}^{L-2} )\backslash \bigcup_{k=0}^{L-2}  B_{Nr/2}(\pmb{\sigma}^k) }  \mkern-10mu \mathbb{P}\left(  \pmb{\sigma}^0 , \pmb{\sigma}^1, \dots , \pmb{\sigma}^{L-1} \in  \mathcal{L}_\varepsilon \right) . 
 \end{align}
 The last probability is upper bounded by 
 \begin{equation}
 \mathbb{P}\left( \sum_{k=0}^{L-1} U_p( \pmb{\sigma}^k) < - N L \varepsilon\right) = \int_{-\infty}^{-NL\varepsilon/  \mathbb{E}\left[ S_{L} ^2\right]}\mkern-10mu \exp\left(-\frac{x^2}{2}\right) \frac{dx}{\sqrt{2\pi} } \leq \exp\left( - \frac{ N^2 L^2 \varepsilon^2}{2 \  \mathbb{E}\left[ S_{L} ^2\right]}\right) , 
 \end{equation}
 where we used the fact that the sum $  S_{L} := \sum_{k=0}^{L-1} U( \pmb{\sigma}^k) $ is a Gaussian random variable with mean zero. Its covariance is bounded by
 \begin{align}
 \mathbb{E}\left[ S_{L} ^2\right]  \ & = L N c_{p,N}(1) + \sum_{k\neq k'}   N c_{p,N}(r_N(\pmb{\sigma}^k, \pmb{\sigma}^{k'}) ) \notag \\
 &  \leq N L \left[  c_{p,N}(1) + (L-1) c_{p,N}\big(1- \tfrac{2}{N} \dist(\pmb{\sigma}^k, \pmb{\sigma}^{k'})\big) \right]  \leq 2 LN \left[1 + L (1-r)^p \right]  .
 \end{align}
 The first inequality is based on the relation $ 1 - r_N(\pmb{\sigma}^k, \pmb{\sigma}^{k'}) = 2  \dist(\pmb{\sigma}^k, \pmb{\sigma}^{k'}) /N \geq r $ for any pair  of sites $ \pmb{\sigma}^k\neq \pmb{\sigma}^{k'} $. The last inequality is a consequence of the assumption~\eqref{eq:ass} and holds for all sufficiently large $ N $. 
 
 Since the volume of the closed ball $ \overline{B_{rN}}(\pmb{\sigma}^j)  $  is estimated in terms of the binary entropy $  \gamma(r ) = - r \ln r - (1-r) \ln (1-r) $ as
 \begin{equation}
  | \overline{B_{rN}} | \leq e^{N \gamma(r) } ,
 \end{equation}
 the right hand side of~\eqref{eq:unionb} is bounded from above by
\begin{align*}
  \exp\left( N \left[ \ln 2 + L \gamma(r) - \frac{L}{1+L\delta_p(r)} \frac{\varepsilon^2}{4}  \right] \right)  \qquad \mbox{with} \quad \delta_p(r) := (1-r)^p . 
 \end{align*}
We now pick
$$
r_p :=  \frac{\ln( 4\beta_c^2/ \varepsilon)}{p},   \quad \mbox{such that} \quad \delta_p(r_p) \leq \exp\left( - p  r_p \right) = \frac{\varepsilon}{4\beta_c^2} ,
$$
such that $ \lim_{p\to \infty } r_p = 0 $ and, hence, $ \liminf_{p\to \infty} \gamma(r_p) = 0 $. For $ p $ large enough, one may thus 
pick an integer $ L_p  \in \mathbb{N} $ such that 
$$
  \frac{1}{\delta_p(r_p)} \left( \frac{\varepsilon}{4 \sqrt{\gamma(r_p)} } - 1 \right) \leq L_p  \leq   \frac{1}{\delta_p(r_p)} \left( \frac{\varepsilon}{2 \sqrt{\gamma(r_p)} } - 1 \right). 
$$
 This ensures that 
 $$
 c_p(\varepsilon) := L_p \left( \frac{\varepsilon^2}{4 (1+L_p \delta_p(r_p))} - \gamma(r_p) \right) - \ln 2 \geq  L_p \sqrt{\gamma(r_p) } \ \frac{\varepsilon}{2} \left( 1 - \frac{2 \sqrt{\gamma(r_p) }}{\varepsilon} \right) - \frac{\beta_c^2}{2} > 0 ,
 $$
 for all $ p $ sufficiently large. 
 Moreover, $ \liminf_{p\to \infty} L_p = \infty $ as well as $ \lim_{p\to \infty } L_p r_p = 0 $ as claimed. 
 \end{proof}

\subsection{Upper bound}
We are now ready to finish the upper bound with an argument similar to the one employed in~\cite{MW20}. 
\begin{proof}[Proof of Theorem~\ref{thm:main}]  Thanks to  Proposition~\ref{prop:lower}, it remains to establish an upper bound. 
Starting from the representation~\eqref{eq:directsumH}, we pick the parameters $ r \equiv r_p(\varepsilon) $ and $ L = L_p(\varepsilon) $ as in Lemma~\ref{lem:prob} and assume the occurrence of the event $ \Omega_{p,N}(r,L,\varepsilon) $ from this Lemma. In this situation, Lemma~\ref{lem:opnormT} yields:
\begin{align}
\tr e^{-\beta H_{p,N }} & \leq e^{\beta \Gamma  \| T_{ \mathcal{L}_\varepsilon^+ }  \|   } \left(  \tr_{\ell^2( \mathcal{L}_\varepsilon)}  e^{-\beta U_p} + \tr_{\ell^2( \mathcal{L}_\varepsilon^c)}  e^{-\beta H_{\mathcal{L}_\varepsilon^c } }  \right)  \notag \\
& \leq e^{ 2 \beta \Gamma N \sqrt{rL}  } \left(  \tr e^{-\beta U_p} +  e^{\beta \varepsilon N } \tr  e^{\beta \Gamma T} \right) .
\end{align}
Here the last step used the fact that $ H_{\mathcal{L}_\varepsilon^c } \geq - N \varepsilon - \Gamma T_{\mathcal{L}_\varepsilon^c } $, and the monotonicity of the partition function in the domain, i.e.\ $  \tr_{\ell^2( \mathcal{L}_\varepsilon^c)} \exp\left( \beta \Gamma T_{\mathcal{L}_\varepsilon^c }  \right) \leq  \tr  e^{\beta \Gamma T}   $, which follows from the non-negativity of the matrix-elements of $ T $.  
 The above bound then implies that on  the event $ \Omega_{p,N}(r,L,\varepsilon) $:
\begin{equation}\label{eq:proof2}
\limsup_{N\to \infty} \Phi_{p,N}(\beta, \Gamma) \leq 2 \beta \Gamma  \sqrt{rL}  + \max\left\{ \limsup_{N\to \infty} \Phi_{p,N}(\beta, 0) , \beta \varepsilon + \ln \cosh(\beta \Gamma) \right\} . 
\end{equation}
By the  self-averaging property~\eqref{eq:selfaveraging} of the classical pressure ($ \Gamma = 0 $), one may further restrict to an event $\widehat \Omega_{p,N}(r,L,\varepsilon) \subset \Omega_{p,N}(r,L,\varepsilon)  $, which still has a probability exponentially close to one as $ N \to \infty $, to conclude that on $ \widehat \Omega_{p,N}(r,L,\varepsilon) $ for both $ \# \in \{0, \Gamma \} $: 
\begin{equation}
\limsup_{N\to \infty} \Phi_{p,N}(\beta, \#) =  \limsup_{N\to \infty} \mathbb{E}\left[ \Phi_{p,N}(\beta, \# )\right]  
\end{equation}
In turn, in the classical case $ \# = 0 $,  Proposition~\ref{prop:Clpinfty} guarantees that in the subsequent limit $ p \to \infty $ the right side converges to the REM's pressure $  \Phi_{\infty}(\beta, 0) $. 
In this limit, $ \lim_{p\to \infty} L_p r_p = 0 $ for any $ \varepsilon > 0 $ by Lemma~\ref{lem:prob}, so that the first term in the right side of~\eqref{eq:proof2} vanishes. 

The proof is completed by a  Borel-Cantelli argument, which relies on the fact  that the probabilities of the complementary event to $  \widehat \Omega_{p,N}(r,L,\varepsilon)  $ are summable in $ N $. 
\end{proof}

\minisec{Acknowledgments}
This work was supported by the DFG under EXC-2111--390814868 and DFG-TRR 352-Project-ID 470903074.

	\bigskip
	\bigskip
	\begin{minipage}{0.6\linewidth}
		\noindent Anouar Kouraich, Simone Warzel\\
		Departments of Mathematics \& Physics, and MCQST \\
		Technische Universit\"{a}t M\"{u}nchen
	\end{minipage}
	\bigskip
	\bigskip
	\begin{minipage}{0.5\linewidth}
		\noindent Chokri Manai \\
		Courant Institute of Mathematical Sciences \\
		New York University
	\end{minipage}

\end{document}